\documentclass[letterpaper, 10 pt, conference]{ieeeconf}  % Comment this line out if you need a4paper

\IEEEoverridecommandlockouts                              % This command is only needed if
\usepackage[english]{babel}
\usepackage[]{graphicx}
\usepackage{amsmath}
\usepackage{graphicx}
\usepackage{color}
\usepackage{graphics} % for pdf, bitmapped graphics files
\usepackage{mathptmx} % assumes new font selection scheme installed
\usepackage{amssymb}  % assumes amsmath package installed
\usepackage{epstopdf}
\usepackage{graphicx}
\usepackage{array}

\newtheorem{remark}{Remark}
\newtheorem{assumption}{Assumption}

\newtheorem{thm}{\textbf{Theorem}}
\graphicspath{{Figures/}}

\title{\LARGE \bf
Robust multi-rate predictive control using multi-step prediction models learned from data
}

\author{Enrico Terzi$^{1}$, Marcello Farina$^{1}$, Lorenzo Fagiano$^{1}$, and Riccardo Scattolini$^{1}$
\thanks{$^{1}$ The authors are with the Dipartimento di Elettronica, Informazione e Bioingegneria, Politecnico di Milano, Via Ponzio 34/5, 20133, Milano, Italy. E-mail: {\tt\small name.surname@polimi.it}}}

\begin{document}

\maketitle
\thispagestyle{empty}
\pagestyle{empty}

\begin{abstract}
This note extends a recently proposed algorithm for model identification and robust model predictive control (MPC) of asymptotically stable, linear time-invariant systems subject to process and measurement disturbances. Independent output predictors for different horizon values are estimated with Set Membership methods. It is shown that the corresponding prediction error bounds are the least conservative in the considered model class. Then, a new multi-rate robust MPC algorithm is developed, employing said multi-step predictors to robustly enforce constraints and stability against disturbances and model uncertainty, and to reduce conservativeness. A simulation example illustrates the effectiveness of the approach.
\end{abstract}

\section{Introduction}\label{Intro}

In a recent paper \cite{TFFS_aut}, we presented a unitary approach to model identification and robust Model Predictive Control (MPC) design for linear, asymptotically stable, discrete time systems subject to process and measurement disturbances. A Set Membership (SM) identification approach was used to obtain multi-step prediction models used in the cost function definition, while state and control constraints were tightened by propagating the uncertainty bound of a simulation model, tuned using the knowledge of the multi-step models and the associated error intervals. Being the multi-step predictors linear in their parameters, it was possible to derive tight uncertainty bounds in a tractable way. However, these bounds were not directly exploited to deal robustly with constraints, with a consequent limited advantage in terms of conservativeness reduction in the constraint tightening procedure.\\
In the present paper, we develop this line of research with two main contributions: first, we prove that the prediction error bounds obtained with the SM approach proposed in \cite{TFFS_aut} are smaller than those of \emph{any} linear simulation model iterated $p$ times. This further motivates the use of such predictors both in the cost function and for constraint tightening. We do so in our second contribution, since we propose a new robust MPC scheme that explicitly relies on the optimal SM multi-step models, thus dramatically reducing conservativeness. To deal with the particular structure of the multi-step predictors, which prevents the use of a standard robust MPC approach, we adopt a novel multi-rate receding horizon strategy, for which we prove guaranteed constraint satisfaction and convergence properties. Many multirate schemes have been proposed in the literature for predictive control design, see for example \cite{scattolini1994multirate,wang2006multirate}, \cite{wang2020robust} and the references therein, usually to cope with different sampling rates in outputs sampling, state update, and control implementation. On the contrary, here the multirate implementation stems from the particular form of the predictors.\\
In the last section of the paper, the new approach is compared with that of \cite{TFFS_aut} in a simulation example. The proofs of the main results are reported in Appendix.\smallskip\\
\textbf{Notation}: $I_n$ is the identity matrix of dimension $n$, $\bar{I}_n$ is the matrix with zero entries except for those on the anti-diagonal, which are equal to 1, $0_{m,n}$ is the null matrix of dimensions $m$ and $n$. The Cartesian product between $n$ sets $\textbf{T}_1,\ldots,\textbf{T}_n$ is $\prod\limits_{i=1}^{n}\textbf{T}_i$. For a generic vector $x$, $\|x\|^2\doteq{x^Tx}$ and $\|x\|^2_Q\doteq{x^TQx}$ with $Q$ being a given square matrix of suitable dimension. For a matrix $A$, $\|A\|=\sup_{x\neq0}\frac{\|Ax\|}{\|x\|}$ is its induced 2-norm and $\rho(M)$ its spectral radius, i.e. the maximum absolute value of its eigenvalues.
Given sets $\textbf{A},\,\textbf{B}\subset\mathbb{R}^n$, $\textbf{A}\oplus\textbf{B}=\{a+b:a\in\textbf{A},\,b\in\textbf{B}\}$ and $\textbf{A}\ominus\textbf{B}=\{a\in\textbf{A}:\forall b\in\textbf{B},\,a+b\in\textbf{A}\}$. %denote the Minkowsky sum \cite{Oks2006} and Pontryagin difference \cite{kvasnica2004multi}, respectively (CERCHEREI DI METTERE UN SOLO RIFERIMENTO BIBLIOGRAFICO PER ENTRAMBE?).
\section{Problem statement, identification algorithm, and error bounds}\label{Statement}
Consider a linear and time-invariant (LTI) discrete-time system of order $n$ with input $u(k)\in\mathbb{R}$, output $z(k)\in\mathbb{R}$, measured output $y(k)\in\mathbb{R}$, process disturbance $v(k)\in\mathbb{R}$, and measurement disturbance $d(k)\in\mathbb{R}$, where $k\in\mathbb{N}$ is the discrete time variable. We define $\varphi^{(p)}_z(k) \in \mathbb{R}^{2n+p-1}$  as:
\begin{align}
\varphi^{(p)}_z(k)=&[z(k),\ldots,z(k-n+1),\\
&u(k-1),\ldots,u(k-n+1),u(k),\dots,u(k+p-1)]^T, \nonumber
\end{align}
with $p \in \mathbb{N}$. The system can be expressed in  ARX (autoregressive-exogenous) form as
\begin{equation}
\begin{cases}
z(k+1)= \bar{\theta}^{(1)^T} \varphi^{(1)}_z(k) + v(k) \\
y(k)=z(k) + d(k),\\
\end{cases}
\label{eq:realsys}
\end{equation}\\
where $\bar{\theta}^{(1)} \in \mathbb{R}^{2n}$ is the vector of unknown parameters.
\begin{assumption}\textit{(Disturbance boundedness)}. \label{A:noise_bound}
$|v(k)|\leq \bar{v}, \quad |d(k)|\leq \bar{d}$, $\quad \forall k \in \mathbb{N}$ with $\bar{d}$ known.\hfill $\square$
\end{assumption}
The value of $\bar{d}$ is assumed to be available from prior knowledge, and/or it can also be estimated from data, see e.g. \cite{LaFa2020}, whereas $\bar{v}$ is not necessarily known.\\
%Note that Assumption \ref{A:noise_bound} is rather standard in the literature, and typically holds in practice.
%
Using the SM method presented in \cite{TFFS_aut}, the following predictors of order $o$ can be obtained for all $p$ values up to a finite horizon $\overline{p}$:
%\begin{equation}
%\begin{array}{lcl}
%\hat{z}(k+p)&=&\hat{\theta}_{AR}^{(p)*^T}\begin{bmatrix}y(k)\\\vdots\\y(k-o+1)\end{bmatrix}+\hat{\theta}_{U}^{(p)*^T}\begin{bmatrix}u(k-1)\\\vdots\\u(k-o+1)\end{bmatrix}\\
%&&+\hat{\theta}_{\bar{U}}^{(p)*^T}\begin{bmatrix}u(k)\\\vdots\\u(k+p-1)\end{bmatrix}, p =1,\dots, \overline{p}
%\end{array}
%\label{eq:p-steps-model}
%\end{equation}
%where $\hat{\theta}_{AR}^{(p)*^T} \in \mathbb{R}^{o}$, $\hat{\theta}_{U}^{(p)*^T} \in \mathbb{R}^{o-1},\hat{\theta}_{\bar{U}}^{(p)*^T} \in \mathbb{R}^p$ are vectors of known parameters resulting from the identification phase. We  refer to these predictors as ``multi-step" in the remainder. For notational convenience, we write \eqref{eq:p-steps-model} in short as
\begin{equation}
\hat{z}(k+p)=\hat{\theta}^{(p)*^T}\varphi_y^{(p)}(k),
\label{eq:predictor_theta_phi}
\end{equation}
where $\hat{\theta}^{(p)*}=[\hat{\theta}_{AR}^{(p)*^T} \quad \hat{\theta}_{U}^{(p)*^T} \quad \hat{\theta}_{\bar{U}}^{(p)*^T}]^T$ and  $\hat{\theta}_{AR}^{(p)*^T} \in \mathbb{R}^{o}$, $\hat{\theta}_{U}^{(p)*^T} \in \mathbb{R}^{o-1},\hat{\theta}_{\bar{U}}^{(p)*^T} \in \mathbb{R}^p$ are vectors of known parameters resulting from the identification phase. We  refer to these predictors as \emph{multi-step} in the remainder. The derivation of  $\hat{\theta}^{(p)*}$ for a given value of $p$ is recalled later on in this section. Moreover, in \eqref{eq:predictor_theta_phi}
\begin{equation*}
\begin{array}{rcl}
\varphi^{(p)}_y(k)&=&[y(k),\ldots,y(k-o+1),\\
&&u(k-1),\ldots,u(k-o+1),u(k),\dots u(k+p-1)]^T
\end{array}
\end{equation*}
\begin{assumption}\textit{(Model order)}\label{a:model_order}
The order of the models \eqref{eq:predictor_theta_phi} is $o \geq n$ \hfill $\square$
\label{ass:modelorder}
\end{assumption}
\smallskip
An algorithm to estimate $o$ is described in \cite{LaFa2020}.
The SM learning phase also returns an estimate of the bound on the worst-case prediction error:
\begin{equation}
|z(k+p)-\hat{z}(k+p)|\leq\hat{\tau}_p(\hat{\theta}^{(p)*})
\label{eq:disturbance_tau}
\end{equation}
In fact, for each step $p\leq\bar{p}$ one can derive a guaranteed upper bound $\hat{\tau}_p$ of the difference between the nominal output and its prediction obtained with a generic predictor
\begin{equation}
\hat{z}(k+p)=\hat{\theta}^{(p)^T}\varphi_y^{(p)}(k)
\label{eq:generic_predictor}
\end{equation}
%\begin{equation}\label{eq:tau_bound}
%|z(k+p)-\hat{z}(k+p)| \leq \hat{\tau}_p
%\end{equation}
%The keypoint is the definition of model sets consistent with the prior information of the system, named Feasible Parameter Sets (FPS), which can be exploited to compute the global guaranteed error bound. The term global indicates that the resulting bound is tight and exploits all the prior available information. The main ingredients needed to suitably define FPS are now recalled, for a more detailed discussion please refer to \cite{TFFS18a}.
%
%
%is the product of a parameter vector $\hat{\theta}^{(p)}$ times the regressor vector $\varphi_y^{(p)}(k)$.
%Note that $\varphi_z^{p}(k)$ contains samples of variable $z$ that is not accessible, so that it is eventually necessary to resort to $\varphi_y^{p}(k)$ containing $y$ samples.
%Also note that, in view of Assumption \ref{ass:modelorder}, system \eqref{eq:realsys} belongs to the considered model class, therefore there exist parameter %vectors $\bar{\theta}^{(p)}$, $\bar{\theta}^{(p)}_v$ and $\bar{\theta}^{(p)}_d$ such that
%\begin{align}
%{z}(k+p)	&=\bar{\theta}^{(p)^T}\varphi^{(p)}_z(k) + \bar{\theta}^{(p)^T}_v\mathcal{V}^{(p)}(k) \notag\\
%			&=\bar{\theta}^{(p)^T}\varphi^{(p)}_y(k)+\bar{\theta}^{(p)^T}_v\mathcal{V}^{(p)}(k)+\bar{\theta}^{(p)^T}_d\mathcal{D}^{(p)}(k)
%\label{eq:predictor_k_p_system}
%\end{align}
%where $\mathcal{V}^{(p)}(k)=\begin{bmatrix}v(k)& \dots & v(k+p-1)\end{bmatrix}^T$ and $\mathcal{D}^{(p)}(k)=\begin{bmatrix}d(k)& \dots & %d(k+p-1)\end{bmatrix}^T$.
For the identification of $\hat{\theta}^{(p)*}$ and $\hat{\tau}_p$ a finite number $N$ of measured data is available, composed of pairs $(\varphi^{(p)}_y(k),\,y(k+p)), \quad k=1,\ldots,N$. 
We first estimate an error bound $\hat{\bar{\epsilon}}_p=\alpha\underline{\lambda}_p$,
$\forall p=1,\dots,\bar{p}$, through
$$\begin{array}{c}
\underline{\lambda}_p=\min\limits_{\theta^{(p)},\lambda}\,\lambda \,,\quad
\text{subject to}\\
|y(k+p)-\theta^{(p)^T}\varphi_y^{(p)}(k)|\leq\lambda+\overline{d}, \text{ } k=1,\dots,N,
\end{array}$$
The latter value is inflated by a scalar $\alpha>1$ to account for the fact that the available dataset is finite.
The Feasible Parameter Sets (FPSs) are then defined as
\begin{equation}
\begin{array}{c}
\Theta^{(p)}=\left\{ \hat{\theta}^{(p)} : |y(k+p)-{\hat{\theta}}^{(p)^T}\varphi^{(p)}_y(k)|\leq\hat{\bar{\epsilon}}_p+\overline{d},\right.\\
\left.k=1,\ldots,N  \right\}
\end{array}
\label{FPSdef}
\end{equation}
For each $p$, $\Theta^{(p)}$ is a convex set and, if the data are informative enough, it is also compact. This property can be checked easily by linear programming; if the set $\Theta^{(p)}$ is not bounded then this is a sign that more informative data should be collected. In the remainder, we consider that $\Theta^{(p)}$ is compact for any $p$.
Let us further denote with $\Phi^{(p)} \subseteq \mathbb{R}^{2o-1+p}$ a compact set containing all possible values of $\varphi^{(p)}_y(k)$. In practice, this means that we restrict our analysis and results to a set of system trajectories of interest, which contains the available data points. This is a reasonable assumption in practice.
Since $\bar{\theta}^{(p)}$ in \eqref{eq:realsys} belongs to $\Theta^{(p)}$, the smallest bound on the error $|z(k+p)-\hat{z}(k+p)|$ (see \eqref{eq:disturbance_tau}) $\forall p=1,\dots,\bar{p}$  is:
\begin{equation}
\begin{array}{l}
\lvert z(k+p) - \hat{z}(k+p) \rvert \\
\leq\max\limits_{\varphi^{(p)}_y\in\Phi^{(p)}}\max\limits_{{\theta}^{(p)} \in\Theta^{(p)}}
\lvert ({\theta}^{(p)}-\hat{\theta}^{(p)})^T\varphi^{(p)}_y \rvert + \hat{\bar{\epsilon}}_p
=\tau_p(\hat{\theta}^{(p)})
\end{array}
\label{eq:tau_p}
\end{equation}
The bound \eqref{eq:tau_p} is global, since it holds for any regressor value inside $\Phi^{(p)}$ and for any model compatible with the data, i.e. contained in the set $\Theta^{(p)}$. However it cannot be computed in practice since the set $\Phi^{(p)}$ is not available. On the other hand, an approximation $\hat{\tau}_p(\hat{\theta}^{(p)})\approx\tau_p(\hat{\theta}^{(p)})$ can be easily computed as $\hat{\tau}_p(\hat{\theta}^{(p)})=\gamma \underline{\tau}_p(\hat{\theta}^{(p)})$ with
\begin{align}\label{tau_estim}
\underline{\tau}_p(\hat{\theta}^{(p)})=\max\limits_{k=1,\ldots,N}\max\limits_{{\theta}^{(p)} \in\Theta^{(p)}}
\lvert ({\theta}^{(p)}-\hat{\theta}^{(p)})^T\varphi^{(p)}_y(k) \rvert + \hat{\bar{\epsilon}}_p
\end{align}
i.e. by computing the worst-case prediction error with respect to the available data.
This approximation includes a second scaling factor $\gamma \geq 1$, again to account for the finite available dataset.
The nominal predictor \eqref{eq:predictor_theta_phi}, for each step $p$, is chosen as the minimizer of this worst case error $\underline{\tau}_p(\hat{\theta}^{(p)})$, i.e.
\begin{equation}
\hat{\theta}^{(p)*}=\arg\min_{\hat{\theta}^{(p)} \in \Theta^{(p)}}\underline{\tau}_p(\hat{\theta}^{(p)})
\label{eq:nominal_model}
\end{equation}
%Solving the optimization problem \eqref{eq:nominal_model} requires the solution to $2N$ linear programs (LPs), as shown in \cite{TFFS_aut}.
%
%\subsection{Optimality of the multi-step error bounds}
The following theorem is concerned with the optimality (in terms of size of the uncertainty bound) of the multistep prediction models.
%states that, for each prediction step $p$, the corresponding optimal multi-step predictor \eqref{eq:p-steps-model} yields a guaranteed error bound that is no larger than the one achieved by any LTI model (i.e. 1-step predictor) iterated $p$ times.
%
\begin{thm}\label{thm1}
Consider any 1-step-ahead LTI system model (i.e. of the form \eqref{eq:generic_predictor} with $p=1$) with coefficient vector $\hat{\theta}^{(1)} \in \mathbb{R}^{2o}$. Let $\hat{\theta}^{(p),1} \in \mathbb{R}^{2o-1+p}$ be the corresponding vector of multi-step predictor coefficients, obtained by iterating $p$ times such a 1-step-ahead model. Then, for all $p=1,\dots,\bar{p}$ it holds:
% the optimal $p$-steps ahead model $\theta_p$ provides a worst-case error $\tau_p$ that is smaller than any iterated $1$ step ahead model, i.e.
\begin{equation}
\tau_p(\hat{\theta}^{(p)*}) \leq \tau_p(\hat{\theta}^{(p),1}).
\end{equation}
\label{eq:thm1}
\end{thm}
\begin{proof}
See the Appendix.
\end{proof}
%
%
%More details on the derivation of $\hat{\theta}^{(p),1}$ as a function $\hat{\theta}^{(1)}$ are given in the proof.
Theorem \ref{thm1} justifies the use of multi-step models for robust MPC design, since in general they yield smaller error bounds.%, as confirmed also in our simulation results of Section \ref{sec:simulations}.
%%Notably, in Theorem \ref{thm1}, $\hat{\theta}^{(1)}$ does not necessarily belong to the set $\Theta^{(1)}$.
%\smallskip
%
\section{MPC design and properties} \label{Formulation}
%
%\subsection{Multi-rate control}
The multi-step models previously introduced can not be directly used in existing robust MPC schemes. Therefore we propose a new multirate MPC approach where  the predicted behavior of the system is optimized by considering a prediction/control horizon of $N_p$ ``long'' steps, with index $j\in\mathbb{N},$ each one consisting of $\bar{p}$ ``short'' sampling times with index $k$. Note that the ``short'' sampling interval is the one assumed for the true system \eqref{eq:realsys}.
The optimal control problem is thus solved at every long step $j$  (i.e. every $\bar{p}$ short steps) and the solution provides the values of the control input to be applied at each step $k$ in the interval $\{j\bar{p},\dots,\,(j+1)\bar{p}-1\}$ according to a standard receding horizon formulation. For  clarity, we represent the long and short sampling times on a common time-scale in Figure \ref{Figure:scales}. %, where $j \in \mathbb{N}$ is the time index related to the long sampling time.
Also, in the remainder we will use the upper-case letters to denote variables defined at a long sampling time.
%\begin{remark}
%The multi-rate control scheme only requires the knowledge of multi-step prediction models and the associated prediction error bounds, not necessarily obtained with the algorithm presented in \cite{TFFS_aut} and recalled in Section \ref{Sec:Bounds}.
%\end{remark}
%
\begin{figure}[thpb]
\centering
\includegraphics[scale=0.3]{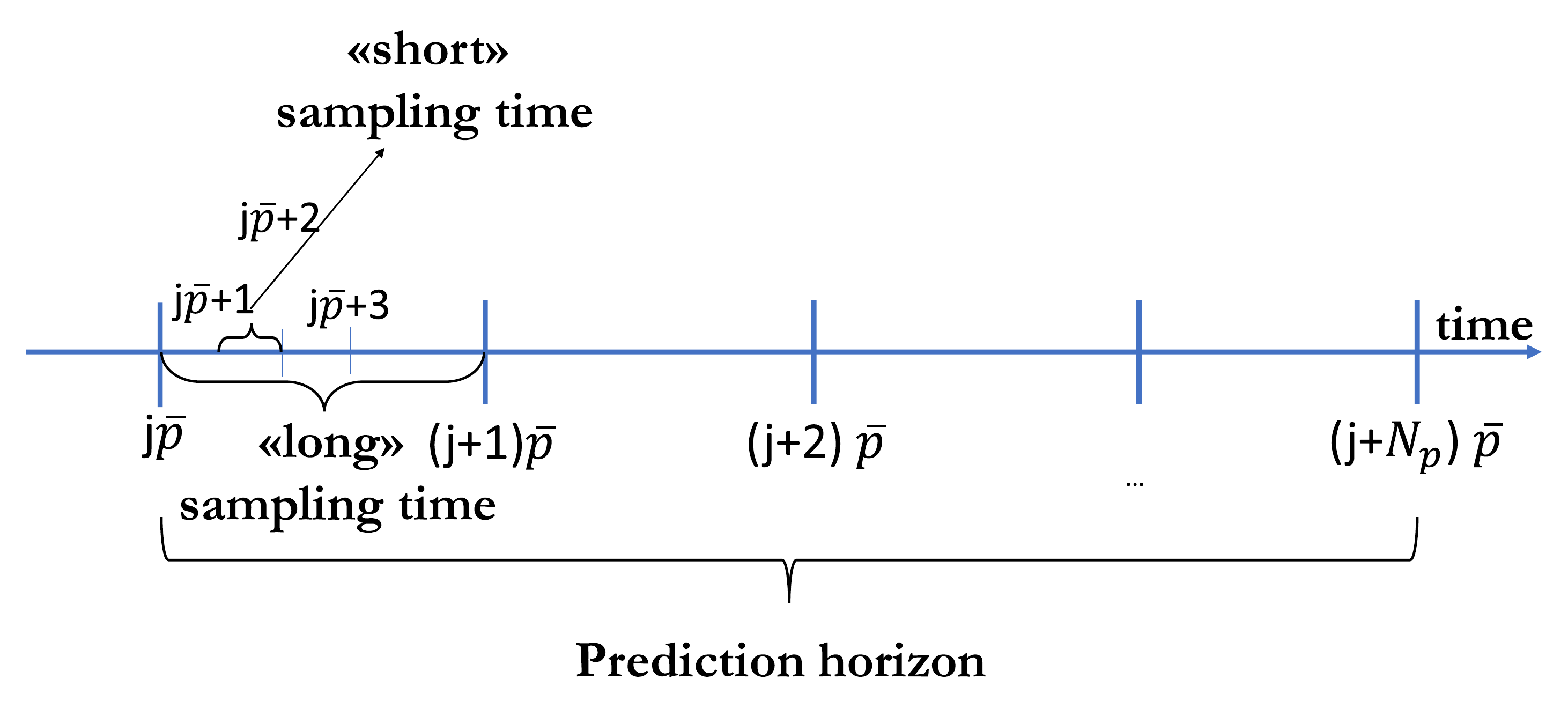}
\caption{Sketch of the time-scales involved in the simulation, with ``long'' and ``short'' sampling times.}
\label{Figure:scales}
\end{figure}
%
%
%\subsection{Derivation of the state space representation and model assumptions}
%To formulate the MPC problem, we resort to a compact state space representation of the identified model, derived based on a $\bar{p}$ steps time-scale.\\
%
Assume $\bar{p}>o$ for simplicity, although it is not necessary, and define the system state, the input, and the disturbance at time $j$ as $X(j)=[y(j\bar{p}), \dots y(j\bar{p}-o+1), u(j\bar{p}-1), \dots, u(j\bar{p}-o+1)]^T$, $U(j)=[u(j\bar{p}), \dots, u(j\bar{p}+\bar{p}-1)]^T$, $W(j)=[w_1(j\bar{p}), \dots, w_{\bar{p}}(j\bar{p})]^T,$ respectively. Denote with $\bar{w}_p$ a value such that $|w_p(j\bar{p})|\leq  \bar{w}_p$, for all $p=1, \dots,\bar{p}$, which accounts for the error stemming from the identification procedure, the process noise, and the measurement disturbance. Given the bound \eqref{tau_estim}, since the state $X(j)$ comprises samples of the measured output  $y$ affected by measurement noise $d$, it is possible to obtain $\bar{w}_p$ as
\begin{equation}
\bar{w}_p=\hat{\tau}_p(\hat{\theta}^{(p)*})+\bar{d}
\label{disturbance_w}
\end{equation}
thus directly exploiting the multi-step error bounds previously obtained.
The state transition equation, that maps the current state $X(j)$ into the $\bar{p}$ steps ahead state $X(j+1)$, is:
\begin{equation}
X(j+1)=\bar{A}X(j)+\bar{B}U(j)+\bar{M}W(j)
\label{eq:state_equation_psteps}
\end{equation}
where:
\begin{equation}
\begin{array}{l}
\bar{A}=\begin{bmatrix}
\hat{\theta}_{AR}^{(\bar{p})*^T} & \hat{\theta}_{U}^{(\bar{p})*^T}\\
\vdots & \vdots\\
\hat{\theta}_{AR}^{(\bar{p}-o+1)*^T} & \hat{\theta}_{U}^{(\bar{p}-o+1)*^T}\\
0_{o-1,o}& 0_{o-1,o-1}
\end{bmatrix},
\bar{B}=\begin{bmatrix}
\hat{\theta}_{\bar{U}}^{(\bar{p})*^T}\\
\vdots\\
\begin{bmatrix} \hat{\theta}_{\bar{U}}^{(\bar{p}-o+1)*^T} & 0_{1,o-1}\end{bmatrix}\\
\begin{bmatrix}
0_{o-1,\bar{p}-o+1} & \bar{I}_{o-1}
\end{bmatrix}
\end{bmatrix}\\
\\
\bar{M}=\begin{bmatrix}
0_{o,\bar{p}-o} & \bar{I}_{o}\\
0_{o-1,\bar{p}-o}& 0_{o-1,o}
\end{bmatrix}
\end{array}
\end{equation}
The following assumption is introduced.
\smallskip
\begin{assumption}
\label{assump:AB}
The pair $(\bar{A},\bar{B})$ is stabilizable.\hfill{}$\square$
\end{assumption}
Since the model is obtained from input-output data, Assumption \ref{assump:AB} is usually satisfied in practice and is thus not restrictive. We rewrite models \eqref{eq:disturbance_tau} as system output equations:
\begin{align}
\hat{z}(j\bar{p}+p)=C_p X(j)+D_p U(j)
\label{eq:predictors}
\end{align}
where
$C_p=\begin{bmatrix}\hat{\theta}_{AR}^{(p)*^T} & \hat{\theta}_{U}^{(p)*^T} \end{bmatrix}$, $D_p=\begin{bmatrix}
\hat{\theta}_{\bar{U}}^{(p)*^T} & 0_{1,\overline{p}-p}\end{bmatrix}$.
\smallskip\\
Consistently with \eqref{eq:disturbance_tau}, we can write
\begin{equation}
z(j\bar{p}+p)=C_p X(j)+D_p U(j)+w_p(j\bar{p})
\label{eq:perturbedy}
\end{equation}
For notational convenience let us stack matrices $C_p$ and $D_p$, for all $p=1,\dots \bar{p}$, as
\begin{equation}
\bar{C}=\begin{bmatrix}C_1^T& \dots &C_{\bar{p}}^T\end{bmatrix}^T, \quad \bar{D}=\begin{bmatrix}D_1^T& \dots&D_{\bar{p}}^T\end{bmatrix}^T
\end{equation}
so that we can define the predictions of outputs in the long sampling time, but at a short sampling period basis, as $\hat{Z}(j)=\begin{bmatrix}\hat{z}(j\bar{p}+1)&\dots& \hat{z}(j\bar{p}+\bar{p}) \end{bmatrix}^T$. Thanks to the predictors in \eqref{eq:predictors}, we write
\begin{equation}
\hat{Z}(j)=\begin{bmatrix}\hat{z}(j\bar{p}+1)\\\vdots\\ \hat{z}(j\bar{p}+\bar{p}) \end{bmatrix}=\bar{C}X(j)+\bar{D}U(j)
\label{eq:Zpiled}
\end{equation}
%
%\subsection{Control design}
%
In the control design phase a tube-based robust control approach is used \cite{mayne2005robust} and the input $U(j)$ is defined as 
\begin{equation}
U(j)=\bar{U}(j)+K(X(j)-\bar{X}(j))
\label{eq:controllaw}
\end{equation}
The input $\bar{U}(j)$ will be computed by MPC, while the term $K(X(j)-\bar{X}(j))$ aims to reduce the error between the state $\bar{X}(j)$ of a suitably defined nominal dynamic system and the actual value of $X(j)$, available at time $k=j\bar{p}$. The gain $K$ is chosen such that $\bar{F}=\bar{A}+\bar{B}K$ is Schur stable, which is possible thanks to Assumption \ref{assump:AB}.\\
The nominal dynamic system is defined based on \eqref{eq:state_equation_psteps}:
\begin{equation}
\bar{X}(j+1)=\bar{A}\bar{X}(j)+\bar{B}\bar{U}(j)
\label{eq:nomstate_eq}
\end{equation}
%where $\bar{U}(j)=\left[\bar{u}(j\bar{p}),\dots,\bar{u}(j\bar{p}+\overline{p}-1) \right]^T$.
%
The $p$ steps ahead nominal output predictor corresponding to \eqref{eq:predictors} is computed as:
\begin{align}
\hat{\bar{Z}}(j)=\bar{C}\bar{X}(j)+\bar{D} \bar{U}(j)
\label{z_predictors_costfcn}
\end{align}
The difference between the real available data vector $X(j)$ and the state of the nominal system is defined as $E(j)=X(j)-\bar{X}(j)$. From \eqref{eq:state_equation_psteps} and \eqref{eq:nomstate_eq}, it evolves according to:
\begin{equation}
E(j+1)=(\bar{A}+\bar{B}K)E(j)+\bar{M}W(j)
\label{eq:errordyn}
\end{equation}
Let $\mathbb{E}$ be a %(minimal, if possible)
robust positively invariant (RPI) \cite{RakovicKouramas2007} set for the system \eqref{eq:errordyn}. Similarly to \cite{mayne2005robust}, the constraints and the optimization problem will be defined with reference to the nominal model \eqref{eq:nomstate_eq}. This will require to define suitable tightened state and input constraints, that allow one to account for the difference between $\bar{X}(j)$ and $X(j)$.
\begin{remark}\label{Rem:comparison}
In \eqref{eq:errordyn} only the last $o$ components of $W(j)$ are involved in the computation of $\mathbb{E}$, and they depend on the estimates $\hat{\tau}_p(\hat{\theta}^{(p)*})$ of the bounds proved to be optimal in Theorem \ref{thm1}, see \eqref{disturbance_w}.
Moreover, since $\bar{A} + \bar{B}K$ is Schur stable and evolves over a (possibly long) $\bar{p}$-steps-ahead period, it is prone to have a smaller spectral radius and norm with respect to the one corresponding to a 1-step state space model, e.g. the one considered in \cite{TFFS_aut}. Thus, this results in a smaller set $\mathbb{E}$ and less conservative constraint tightening, as also illustrated in the example of Section \ref{sec:simulations}.\hfill{}$\square$
\end{remark}
%
%
%\subsection{Constraints}
The MPC controller must guarantee the fulfillment of input and output constraints for all $k\geq 0$:
\begin{equation}
\label{eq:constraints_YU}
u(k)\in\mathbb{U} \quad , \quad z(k) \in \mathbb{Z}
\end{equation}
where $\mathbb{U}$ and $\mathbb{Z}$ are suitable convex sets containing the origin in their interior.
For ease of notation, let us introduce the higher-dimensional convex sets $\mathbb{\textbf{U}}=\mathbb{U}^{\bar{p}}$ and $\mathbb{\textbf{Z}}=\mathbb{Z}^{\bar{p}}$. Similarly to \cite{mayne2005robust}, it is first necessary to constrain $\bar{X}(j)$ at time $j\bar{p}$ to lie in the neighborhood of $X(j)$, i.e
\begin{subequations}
\begin{align}
X(j)-\bar{X}(j) \in \mathbb{E}
\end{align}
%As discussed, in the MPC optimization problem, the constraints will be defined only with reference to the state-space model \eqref{eq:nomstate_eq}.
Regarding the input variable, to guarantee that \eqref{eq:constraints_YU} holds from time $j\bar{p}$ to $(j+N_p-1)\bar{p}$, it is enough to enforce the following tightened constraints, for all $i=0,\dots,N_p-1$.
\label{eq:constraints}
\begin{align}
\bar{U}(j+i) \in \mathbb{\textbf{U}}\ominus K\mathbb{E}\label{constraint:Utightened}
\end{align}
As for the output, to guarantee that \eqref{eq:constraints_YU} holds at time $j\bar{p}+1,\dots,(j+N_p)\bar{p}$, we define $\forall p=1,\dots, \bar{p}$
$$\mathbb{T}_p=\{t \in \mathbb{R}: |t|\leq \hat{\tau}_p(\hat{\theta}^{(p)*})\}$$ and the tightened set $\hat{\mathbb{\mathbf{Z}}}$ as
\begin{equation}
\hat{\mathbb{\mathbf{Z}}}=\mathbb{\mathbf{Z}} \ominus \prod_{p=1}^{\bar{p}}\mathbb{T}_p
\end{equation}
This set is such that, by construction, if $\hat{Z}(j+i)\in \hat{\mathbb{\mathbf{Z}}}$, then $Z(j+i)\in \mathbb{\mathbf{Z}}, \quad i=0,\dots,N_p-1$.
We thus enforce the following tightened constraint, again related to the nominal system \eqref{eq:nomstate_eq}, for all $i=0,\dots,N_p-1$.\begin{align}
\hat{\bar{Z}}(j+i)\in \hat{\mathbb{\mathbf{Z}}}\ominus (\bar{C}+\bar{D}K)\mathbb{E}
\end{align}

Finally, to guarantee recursive feasibility, we also need to enforce a terminal constraint of the type
\begin{equation}
\bar{X}(j+N_p)\in\mathbb{X}_F
\end{equation}
where $\mathbb{X}_F$ is defined as a positively invariant set for the system $\hat{X}(j+1)=(\bar{A}+\bar{B}K)\hat{X}(j)$ that verifies
\begin{itemize}
\item $(\bar{C} + \bar{D}K)\mathbb{X}_F \subseteq \hat{\mathbb{\mathbf{Z}}}\ominus (\bar{C}+\bar{D}K)\mathbb{E}$
\item $K\mathbb{X}_F\subseteq \mathbb{\mathbf{U}}\ominus K\mathbb{E}$
\end{itemize}
\end{subequations}
For consistency, the following assumption is required.
\begin{assumption}
\label{assump:sets}
There exists a ball $\mathcal{B}$ in space $\mathbb{R}^{\bar{p}}$, centered at the origin and with radius $\varepsilon$, such that
\begin{subequations}
\begin{align}
(\bar{C}+\bar{D}K)\mathbb{E}\oplus\mathcal{B} &\subseteq \hat{\mathbb{\mathbf{Z}}}\\
K\mathbb{E} \oplus\mathcal{B}& \subseteq \mathbb{\mathbf{U}}
\end{align}\hfill{}$\square$
\end{subequations}
\label{hyp:setcont}
\end{assumption}
%
%Notably, Assumption \ref{hyp:setcont} can be regarded to as a minimal requirement with respect to the uncertainty associated to the identification procedure.
%In fact, the magnitude of set $\mathbb{E}$ is directly proportional to the magnitude of the bounds $\bar{w}_p$, $p=\bar{p}-o+1,\dots,\bar{p}$ of the latest $o$ predictors in time. In case Assumption \ref{hyp:setcont} is not satisfied, it might be necessary either to step back to the identification stage and, if possible, reduce the error bounds (for example by taking more informative data), or to enlarge the constraint sets, since the current ones would be too tight, with respect to the available knowledge of the system, to guarantee constraint satisfaction.
%
%%
%\subsection{Cost function}
%
The cost function to be minimized at time step $k$ is
$$J(j)=\sum_{i=0}^{N_p-1}\|\hat{\bar{Z}}(j+i)\|^2_{Q}+\|\bar{U}(j+i)\|^2_{R}+\|\bar{X}(j+N_p)\|^2_P$$
where $Q=\text{diag}(q_1,\dots,q_{\bar{p}})>0$, $R=\text{diag}(r_0,\dots,r_{\bar{p}-1})>0$, $N_p$ is the prediction horizon, and $P$ is the unique positive definite solution to the Riccati equation (see Assumption \ref{assump:AB})
\begin{equation}
\bar{F}^TP\bar{F}-P=-\left( \bar{G}^TQ\bar{G}+K^TRK \right)
\label{eq:Riccati}
\end{equation}
\noindent where $\bar{G}=(\bar{C}+\bar{D}K)$. %Note that the existence of a solution to \eqref{eq:Riccati} is ensured by Assumption \ref{assump:AB}.
Note that $Q$ and $R$ can be chosen freely while in \cite{TFFS_aut} they were selected according to the solution to an LMI problem, so limiting the possible trade-offs between bandwidth and control activity of the closed-loop system.\\
Now, denoting the vector of decision variables with
$$\mathbf{\bar{U}}(j)=\begin{bmatrix}\bar{U}(j)^T& \dots &\bar{U}(j+N_p-1)^T\end{bmatrix}^T,$$
the optimization problem to be solved at each ``long'' sampling time $j\geq 0$, reads
\begin{align}
J^*(j)=&\min_{\bar{X}(j),\mathbf{\bar{U}}(j)}J(j)
\text{ s.t. constraints } \eqref{eq:constraints}\label{eq:optprb}
\end{align}
If problem \eqref{eq:optprb} is feasible, its solution is denoted with
$\bar{X}^*(j), \mathbf{\bar{U}}^*(j)=[\bar{U}^*(j)^T,\dots,\bar{U}^*(j+N_p-1)^T]^T$, and the input sequence $U^*(j)=\bar{U}^*(j)+K(X(j)-\bar{X}^*(j))$ in \eqref{eq:controllaw} is applied to the system according to the Receding Horizon principle. Also, we denote with  $\bar{X}^*(j+i)$ the future nominal state predictions generated using \eqref{eq:nomstate_eq} with input $\mathbf{\bar{U}}^*(j)$, as well as all the other derived quantities, such as $\hat{\bar{Z}}^*(j)$ (see \eqref{z_predictors_costfcn}).
%The following result states that the described MPC approach yields guaranteed robust constraint satisfaction and asymptotic stability.
%
%
\begin{thm}
\label{res1}
If \eqref{eq:optprb} is feasible at time step $j=0$ then it is feasible at all time steps $j>0$ and, for all $j\geq 0$, the constraints \eqref{eq:constraints_YU} are satisfied. Moreover, $\hat{\bar{Z}}^*(j) \to 0$ as $j \to \infty$. Finally, $\delta(Z(j),(\bar{C}+\bar{D}K)\mathbb{E}) \to 0$ as $j \to \infty$, where $\delta (\alpha,\beta)$ denotes the distance between point $\alpha$ and set $\beta$\hfill{}$\square$.
\end{thm}

\smallskip

\begin{proof}
	See the Appendix.
\end{proof}

\section{Simulation example}\label{sec:simulations}

Consider the system employed in \cite{TFFS_aut}, obtained by discretizing, with sampling time $T_s=0.1$, the continuous-time transfer function
\[
G(s)=\frac{160}{(s+10)(s^2+1.6s+16)}
\]
A dataset of 1000 pairs $(u,y)$ has been collected by exciting the system with a signal $u$ taking value in $\{-1,0,1\}$ randomly each $5$ units of time, and adding the disturbance $v(k)$ and $d(k)$, with $\bar{v}=0.01$ and $\bar{d}=0.1$, respectively, consistently with \eqref{eq:realsys}.
The multi-step bounds estimates $\hat{\tau}_p(\hat{\theta}^{(p)*})$ have been computed according to the algorithm described in \cite{TFFS_aut}, with $\bar{p}=10$ (resulting in a ``long" sampling time equal to $T_s\bar{p}=1 s$) and model order $o=4$. In Figure \ref{Fig_bounds} they are plotted and compared with the bounds computed by simply iterating the simulation model (i.e., the $1$-step ahead predictor) and propagating its uncertainty bound accordingly.

In the control design phase, the matrix $K$ has been computed  with Linear Quadratic (LQ) control, while the prediction horizon for the MPC controller is $N_p=3$. The weighting matrices are defined as $Q=100I_{\bar{p}}$ and $R=1I_{\bar{p}}$, while matrix $P$ is obtained thanks to \eqref{eq:Riccati}.
Both the input $u$ and the output $z$ have been enforced to belong to the set $[-10,10]$ for each time instant.

The input and output trajectories, comparing the closed-loop with the open-loop response of the system, are plotted in Figures \ref{Figure:u} and \ref{Figure:z} together with the relevant bounds. The controller, based on the identified model, is able to regulate the real system \eqref{eq:realsys} to zero with a much faster time constant and sensibly damping the oscillations.
In Table \ref{table:comparison} we also report, for the same tuning of the LQ problem, the spectral radius and norm of the state transition matrix of the nominal system \eqref{eq:nomstate_eq} subject to the auxiliary law $K$, see also Remark \ref{Rem:comparison}. Note that the norm of such matrix directly affects the computation of the invariant set $\mathbb{E}$. Moreover, by comparing the effect on the constraint tightening, we note that, while in \cite{TFFS_aut} the tightened output constraints correspond to the interval $[-7.7,7.7]$ for each prediction step and the input constraints to the interval $[-9.05,9.05]$, with the new algorithm proposed. here we obtain the following box-inequalities, to be intended entry-wise, $i=0,\dots,N_p-1$:
Specifically, define

\scalebox{0.75}{\parbox{\linewidth}{\begin{eqnarray*}
   Z_m &=& \begin{bmatrix}
8.3 \quad 7.4 \quad  7.8 \quad  8.2 \quad  8.8 \quad  9.0 \quad  9.0 \quad  9.3 \quad  9.1 \quad  8.9
\end{bmatrix} \\
  U_m &=& \begin{bmatrix}
9.77 \quad  9.68 \quad  9.72 \quad  9.60 \quad  9.53 \quad  9.60 \quad  9.67 \quad  9.88 \quad  9.87 \quad  9.92
\end{bmatrix}
\end{eqnarray*}
}}

%\scalebox{0.9}{\parbox{\linewidth}{
%\begin{equation}Z_m=\begin{bmatrix}
%8.3 \quad 7.4 \quad  7.8 \quad  8.2 \quad  8.8 \quad  9.0 \quad  9.0 \quad  9.3 \quad  9.1 \quad  8.9
%\end{bmatrix} \text{ and }
%U_m=\begin{bmatrix}
%9.77 \quad  9.68 \quad  9.72 \quad  9.60 \quad  9.53 \quad  9.60 \quad  9.67 \quad  9.88 \quad  9.87 \quad  9.92
%\end{bmatrix}
%\end{equation}
%}}
and the constraints\\
\begin{equation*}
  -Z_m^T\leq \hat{\bar{Z}}(j+i)\leq Z_m^T \text{ and } -U_m^T\leq \bar{U}(j+i)\leq U_m^T
\end{equation*}

%\scalebox{0.9}{\parbox{\linewidth}{
%\begin{equation}\begin{bmatrix}
%-8.3\\-7.4\\-7.8\\-8.2\\-8.8\\-9.0\\-9.0\\-9.3\\-9.1\\-8.9
%\end{bmatrix}\leq \hat{\bar{Z}}(j+i) \leq \begin{bmatrix}
%8.3\\7.4\\7.8\\8.2\\8.8\\9.0\\9.0\\9.3\\9.1\\8.9
%\end{bmatrix} \text{ and }
%\begin{bmatrix}
%-9.77\\-9.68\\-9.72\\-9.60\\-9.53\\-9.60\\-9.67\\-9.88\\-9.87\\-9.92
%\end{bmatrix}\leq \bar{U}(j+i) \leq \begin{bmatrix}
%9.77\\9.68\\9.72\\9.60\\9.53\\9.60\\9.67\\9.88\\9.87\\9.92
%\end{bmatrix}
%\end{equation}
%}}

\noindent which confirm a conservativeness reduction.
%Finally, a different and tailored tuning of the auxiliary law $K$ may bring even better results in the trade-off between input and output constraint tightening.

\large

\begin{table}[h]
	\caption{Table of comparison of radius and spectral norm of state transition matrix}
	\label{table:comparison}
\centering
\scalebox{1.25}{
\begin{tabular}{cc}
	Algorithm in \cite{TFFS_aut} & Proposed algorithm
	\\
	\hline
	\\
	$\rho(A+B_1K)=0.78$ 	& $\rho(\bar{A}+\bar{B}K)=0.2974$\\
	$\|A+B_1K\|=1.77$ 		& $\|\bar{A}+\bar{B}K\|=0.455$
\end{tabular}}
\end{table}

\normalsize

\begin{figure}[thpb]
\centering
\includegraphics[scale=0.45]{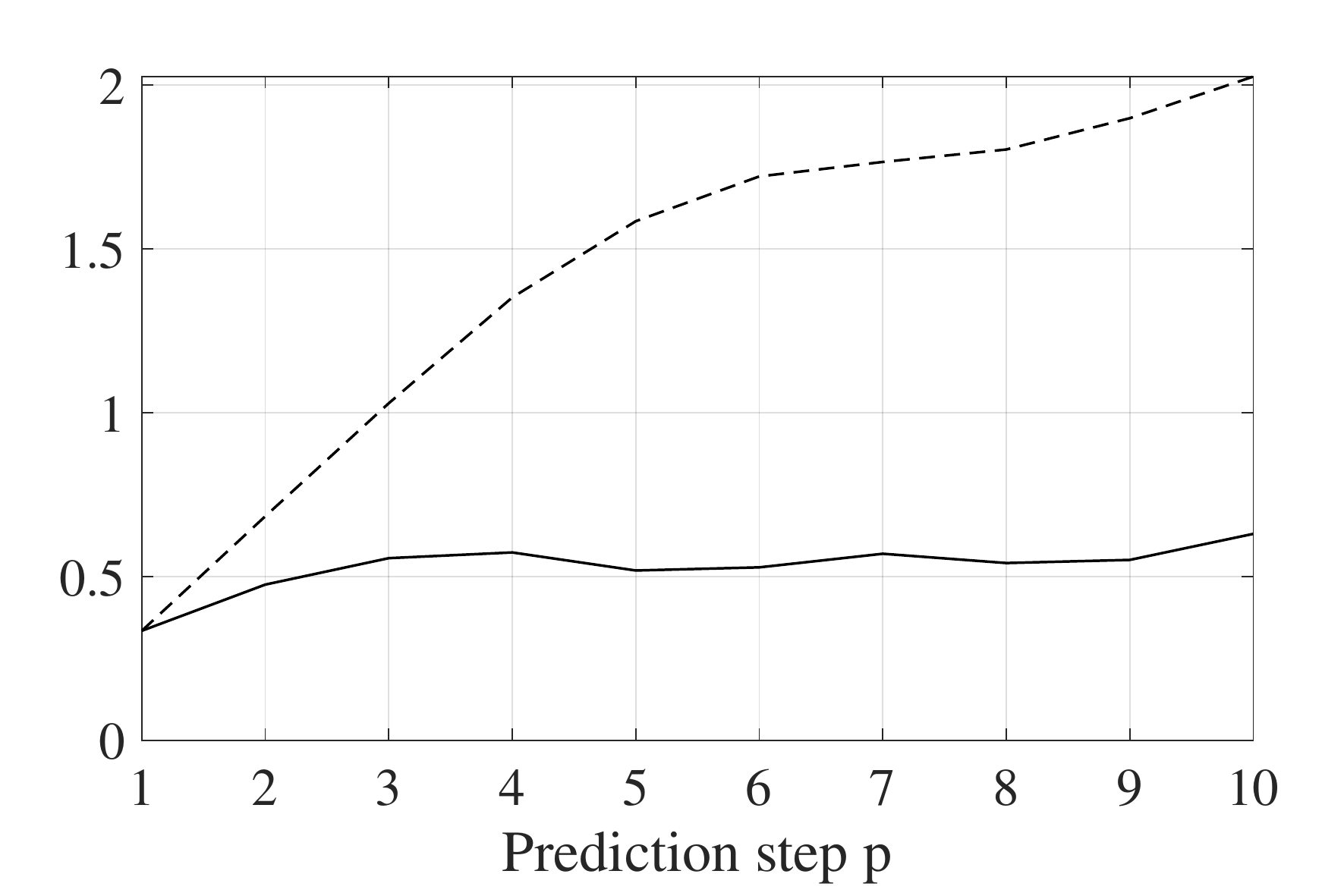}
\caption{Computed bounds. Dashed line: bound obtained by iterating $\bar{w_1}$ with the one-step model, solid line: bounds $\bar{w}_p, p=1,\dots,\bar{p}$ } \label{Fig_bounds}
\end{figure}

\begin{figure}[thpb]
\centering
\includegraphics[scale=0.45]{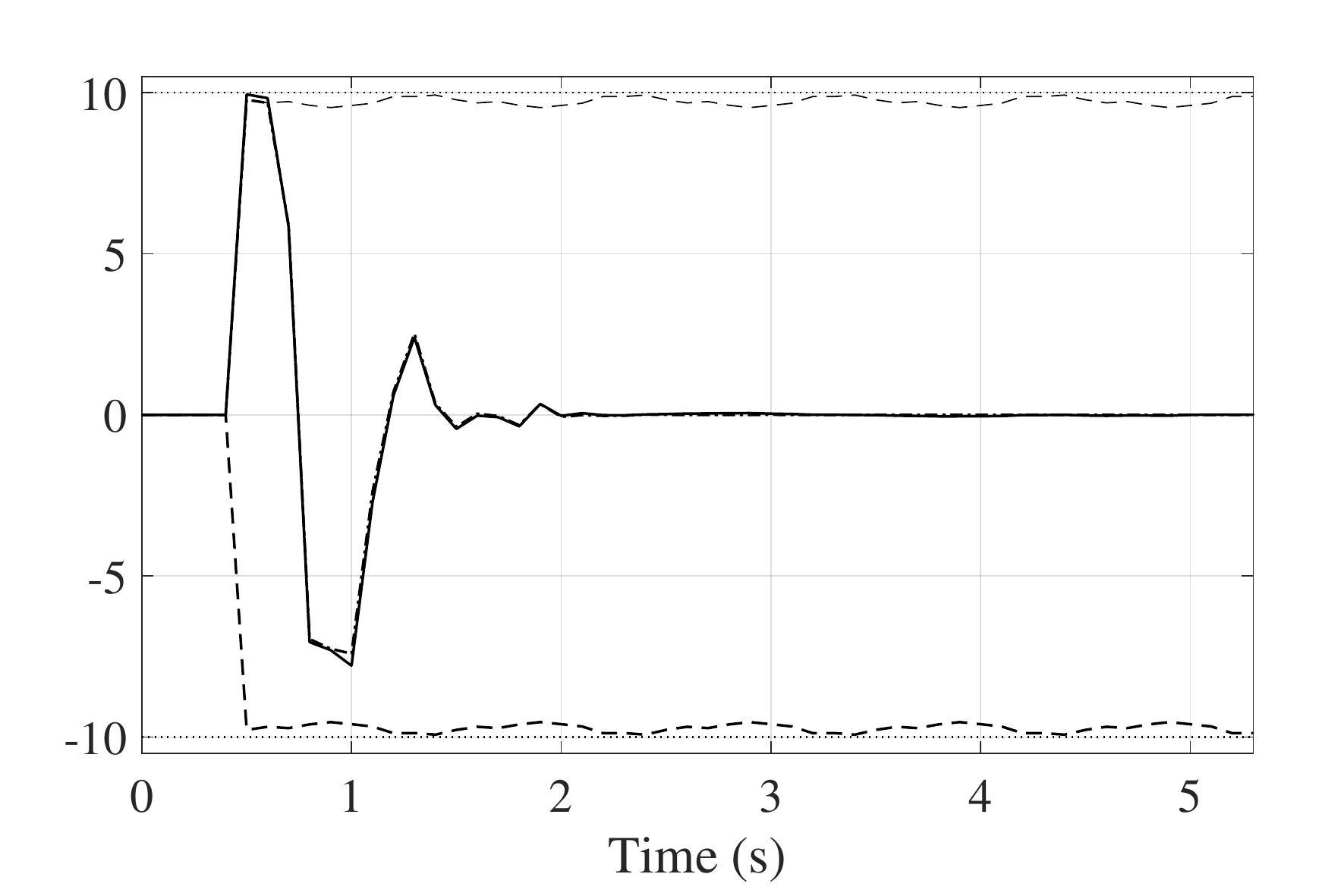}
\caption{Input variable. Dash-dotted line: $\bar{U}(k)$, solid line: $U(k)$, dashed lines: tightened constraints \eqref{constraint:Utightened}, dotted lines: absolute constraints. \label{Figure:u}}
\end{figure}

\begin{figure}[thpb]
\centering
\includegraphics[scale=0.45]{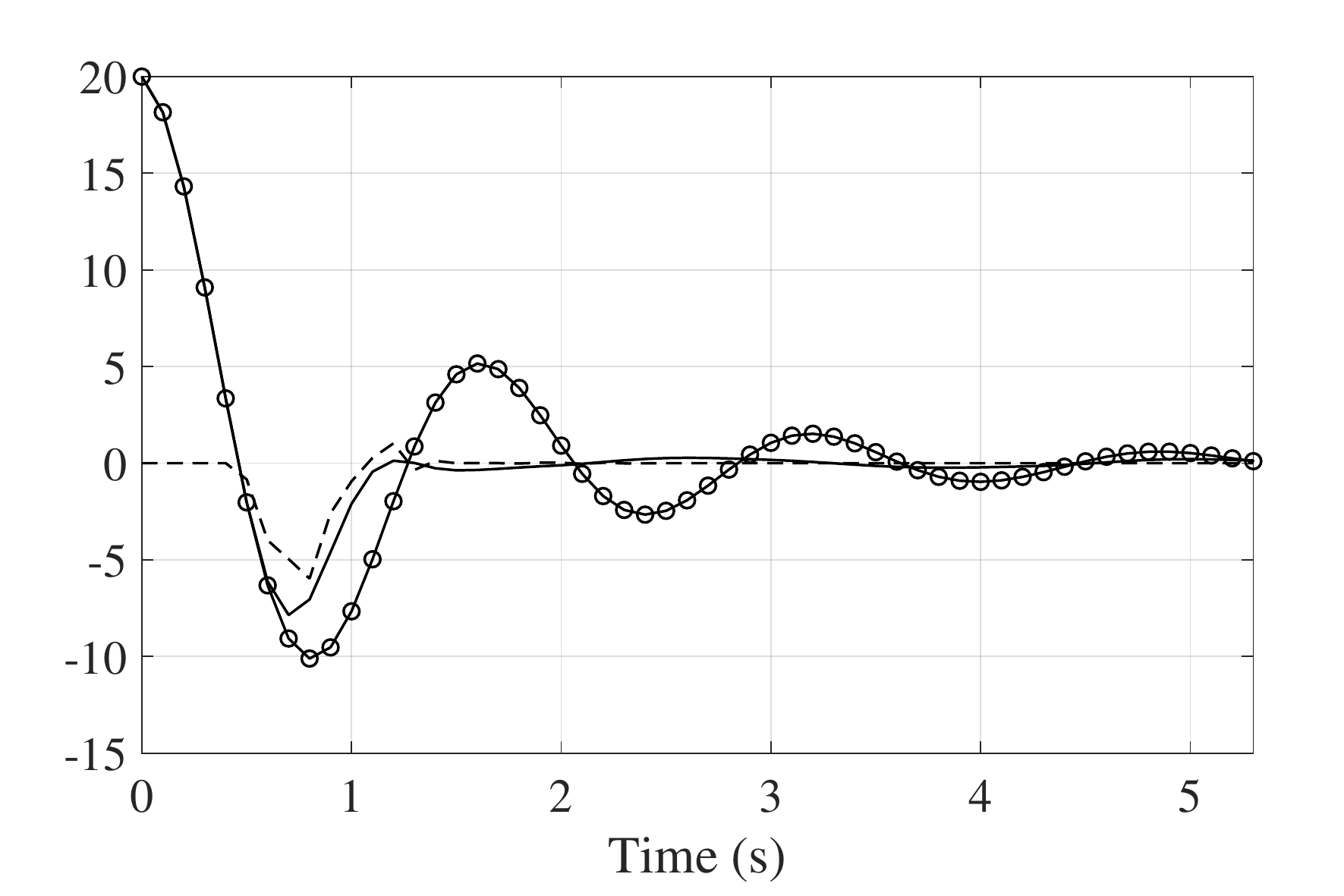}
\caption{Output variable. Solid line: $z(k)$, dashed line: $\hat{\bar{Z}}(j)$, line with circles: open loop response. \label{Figure:z}}
\end{figure}

\section*{Appendix}

\noindent \emph{Proof of Theorem \ref{thm1}}.\\

We derive a multi-step predictor by iterating a generic (simulation, i.e. $1$-step ahead) model with coefficient vector $\hat{\theta}^{(1)}$ and  focusing on the function linking prediction steps $p=1$ and $p=2$. The extension up to $\bar{p}$ is straightforward.
In the following we will use $'$ on the variables predicted with the (possibly iterated) simulation model.\\
First recall the definition of the one-step predictor regressor vector $\varphi_y^{(1)}(k)=\begin{bmatrix}y(k) &  \dots & y(k-o+1) & u(k-1) & \dots & u(k-o+1) & u(k) \end{bmatrix}^T$ and note that $\varphi_y^{(2)}(k)= \begin{bmatrix} \varphi_y^{(1)}(k) \\ u(k+1) \end{bmatrix}$.\\
Assuming to be at time k+1, to proceed 1-step ahead we would need $y(k+1)$ to compute
\begin{equation}
\hat{z}'(k+2)=\hat{\theta}^{(1)^T}\varphi_y^{(1)}(k+1)
\end{equation}
If we are at time $k$, the measured value of $y(k+1)$ is not available, hence its nominal prediction computed with the simulation model is used in its place, i.e. $\hat{z}'(k+1)=\hat{\theta}^{(1)^T}\varphi_y^{(1)}(k)$.
This results in
\begin{equation}\hat{z}'(k+2)=\hat{\theta}^{(1)^T}\varphi_y^{(1)'}(k+1)
\label{eq:1stepiter}
\end{equation}
where $\varphi_y^{(1)'}(k+1)=[\hat{z}'(k+1)\,\dots\, y(k-o+2)\,u(k-1)\,\dots \, u(k-o+2),u(k), u(k+1)]^T$, which can be expressed as a function of $\varphi_y^{(1)}(k)$ as $\varphi_y^{(1)'}(k+1)=$
\begin{equation}
=\begin{bmatrix}
0 \\ \vdots \\ 0 \\ u(k+1)
\end{bmatrix} +
\begin{bmatrix}
&\hat{\theta}^{(1)*^T}  &\\
I_{o-1} 		& 0_{o-1,1} 		& 0_{o-1,o}		\\
0_{1,o-1}		& 0 				& \begin{bmatrix} 0_{1,o-1}			& 1				\end{bmatrix} \\
0_{o-2,o-1} 	&0_{o-2,1}			& \begin{bmatrix} I_{o-2} 			& 0_{o-2,2}		\end{bmatrix}	\\
0_{1,o-1}		&0					& 0_{1,o}
\end{bmatrix}\varphi_y^{(1)}(k)
\label{eq:phik1phik}
\end{equation}
or, in shorter notation $\varphi_y^{(1)'}(k+1)=\mathcal{U}_1+S(\hat{\theta}^{(1)})\varphi_y^{(1)}(k)$, where $\mathcal{U}_1$ and $S(\hat{\theta}^{(1)})$ are are implicitly defined in \eqref{eq:phik1phik}.

By replacing \eqref{eq:phik1phik} in \eqref{eq:1stepiter} we get
\begin{equation}
\begin{array}{ll}
\hat{z}'(k+2)&=\hat{\theta}^{(1)^T} \left( \mathcal{U}_1+S(\hat{\theta}^{(1)})\varphi_y^{(1)}(k) \right)\\
&=\underbrace{\begin{bmatrix}
	\hat{\theta}^{(1)^T}S(\hat{\theta}^{(1)})  & \hat{\theta}^{(1)}_{last}
	\end{bmatrix}}_{\hat{\theta}^{(2),1}}
\begin{bmatrix}
\varphi_y^{(1)}(k) \\ u(k+1)
\end{bmatrix}\\
&=\hat{\theta}^{(2),1}\varphi_y^{(2)}(k)
\end{array}
\label{eq:iteration_1step}
\end{equation}
where $\hat{\theta}^{(1)}_{last}$ is the last element of $\hat{\theta}^{(1)}$, and $ \hat{\theta}^{(2),1}$ is introduced. Note that the entries of $\hat{\theta}^{(2),1}$ are polynomial combinations of the ones of $\hat{\theta}^{(1)}$, see \eqref{eq:iteration_1step}.
With similar manipulations, for any step $p$ the predictor obtained by iterating the simulation model with parameters $\hat{\theta}^{(1)}$ reads
\begin{equation}
\hat{z}'(k+p)=\hat{\theta}^{(p),1^T}\varphi_y^{(p)}(k)
\end{equation}
where $\hat{\theta}^{(p),1}(\hat{\theta}^{(1)}):\mathbb{R}^{2o}\rightarrow \mathbb{R}^{2o+p-1}  $ is a vector function of polynomials of degree up to $p$ of the elements of $\hat{\theta}^{(1)}$.
\medskip

Let us now focus on the worst case prediction error, with arguments similar to \eqref{eq:tau_p}:
\begin{equation}
\begin{array}{ll}
&|z(k+p)-\hat{z}'(k+p)| \leq\\
&\leq |\varphi_y^{(p)^T}(\bar{\theta}^{(p)} - \hat{\theta}^{(p),1}) | + \bar{\epsilon}_p \\
%& \leq \max\limits_{\varphi_y^{(p)} \in \Phi^{(p)}} |\varphi_y^{(p)^T}(\bar{\theta}^{(p)} - \hat{\theta}^{(p),1})| + \bar{\epsilon}_p \\
&\leq \max\limits_{\theta \in \Theta^{(p)}} \max\limits_{\varphi_y^{(p)} \in \Phi^{(p)}} |\varphi_y^{(p)^T}(\theta - \hat{\theta}^{(p),1})| + \bar{\epsilon}_p = \tau_p(\hat{\theta}^{(p),1})
\end{array}
\label{eq:theta'inTheta}
\end{equation}

where the last inequality holds thanks to the fact that $\bar{\theta}^{(p)}\subseteq \Theta^{(p)}$.
We now aim to show that $\tau_p(\hat{\theta}^{(p),1}) \geq \tau_p(\hat{\theta}^{(p)*})$.

For a given vector $\hat{\theta}^{(p),1}$, one of these two cases occur:
\begin{itemize}
	\item If $\hat{\theta}^{(p),1} \in \Theta^{(p)}$, then from \eqref{eq:theta'inTheta}, it follows
	\begin{equation*}
	\begin{split}
	\tau_p(\hat{\theta}^{(p)*})&=\min\limits_{\bar{\theta}^{(p)} \in \Theta^{(p)}}\max\limits_{\theta \in \Theta^{(p)}} \max\limits_{\varphi_y^{(p)} \in \Phi^{(p)}} |\varphi_y^{(p)^T}(\theta-\bar{\theta}^{(p)}) | + \bar{\epsilon}_p\\
	& \leq \tau_p(\hat{\theta}^{(p),1})
	\end{split}
	\end{equation*}
	hence proving the claim.
	\smallskip
	\item If $\hat{\theta}^{(p),1} \notin \Theta^{(p)}$ %(see Figure~\ref{Figure:sketch of the proof})
	, let us consider a generic element $\theta \in \Theta^{(p)}$ and the convex combination
	\begin{equation}
	(1-\bar{\alpha}) \theta + \bar{\alpha} \hat{\theta}^{(p),1} =\theta^{(p)''}
	\label{eq:theta_ii}
	\end{equation} 
	where
	\begin{equation*}
	\begin{array}{cl}
	\bar{\alpha}=& \max\limits_{\alpha} \alpha\\
	s.t.&\alpha \in [0,1]\\
	&(1-\alpha)\theta+\alpha \hat{\theta}^{(p),1} \in \Theta^{(p)}\\
	\end{array}
	\label{eq:theta_p'''}
	\end{equation*}
	The point $\theta^{(p)''}$, belongs to the boundary of $\Theta^{(p)}$ along the direction connecting the chosen $\theta \in \Theta^{(p)}$ to $\hat{\theta}^{(p),1} \notin \Theta^{(p)}$.
	%
	%\begin{figure}[h!]
	%\centering
	%\includegraphics[scale=0.2]{thetas_FPS_img}
	%\caption{Representation of the different model parameters involved in the proof of Theorem \ref{thm1} - case $\hat{\theta}^{(p),1} \notin \Theta^{(p)}$}
	%\label{Figure:sketch of the proof}
	%\end{figure}
	%
	Consider \eqref{eq:theta'inTheta}, omitting the dependence on time $k$ for brevity we compute
	\begin{equation}
	\begin{array}{ll}
	&|\varphi_y^{(p)^T}(\theta-\hat{\theta}^{(p),1})| =\\
	%& = |\varphi_y^{(p)^T} [(1- \bar{\alpha})(\theta - \hat{\theta}^{(p),1}) + \bar{\alpha} (\theta -\hat{\theta}^{(p),1})]|\\
	& = |\varphi_y^{(p)^T}(1-\bar{\alpha})(\theta - \hat{\theta}^{(p),1})| + |\varphi_y^{(p)^T}\bar{\alpha} (\theta - \hat{\theta}^{(p),1})|\\
	& = |\varphi_y^{(p)^T} (\theta^{(p)''}-\hat{\theta}^{(p),1})| + |\varphi_y^{(p)^T}(\theta - \theta^{(p)''})|\\
	\end{array}
	\end{equation}
	where in the last equality the term $\theta^{(p)''}$ as defined in \eqref{eq:theta_ii} has been substituted.
	The latter expression allows us to split the value $$|\varphi_y^{(p)^T}(\theta-\hat{\theta}^{(p),1})|$$ in a contribution given by a predictor $\theta$ inside $\Theta^{(p)}$, namely $|\varphi_y^{(p)^T}(\theta - \theta^{(p)''})|$, plus a contribution outside $\Theta^{(p)}$, that is $|\varphi_y^{(p)^T} (\theta^{(p)''}-\hat{\theta}^{(p),1})|$. Therefore we can write $\tau_p(\hat{\theta}^{(p),1})=$
	\scalebox{0.85}{\parbox{\linewidth}{
			\begin{align}
			&=\max\limits_{\theta \in \Theta^{(p)}} \max\limits_{\varphi_y^{(p)} \in \Phi^{(p)}} |\varphi_y^{(p)^T} (\theta- \hat{\theta}^{(p),1})| + \bar{\epsilon}_p \notag\\
			% &=\max\limits_{\theta \in \Theta^{(p)}} \max\limits_{\varphi_y^{(p)} \in \Phi^{(p)}} \biggl(|\varphi_y^{(p)^T}(\theta - \theta^{(p)''})|+ |\varphi_y^{(p)^T} (\theta^{(p)''}-\hat{\theta}^{(p),1})|\biggr) + \bar{\epsilon}_p  \notag\\
			% &\geq \max\limits_{\theta \in \Theta^{(p)}} \max\limits_{\varphi_y^{(p)} \in \Phi^{(p)}} |\varphi_y^{(p)^T}(\theta - \theta^{(p)''})| + \bar{\epsilon}_p \notag\\
			&\geq \min\limits_{\theta^{(p)''} \in \Theta^{(p)}}\max\limits_{\theta \in \Theta^{(p)}} \max\limits_{\varphi_y^{(p)} \in \Phi^{(p)}} |\varphi_y^{(p)^T}(\theta - \theta^{(p)''})| + \bar{\epsilon}_p =\tau_p(\hat{\theta}^{(p)*})
			\label{eq:min_tau}
			\end{align}
	}}
	that completes the proof. Note that in \eqref{eq:min_tau}, for a given $\hat{\theta}^{(p),1}$,  $\theta^{(p)''}$ depends only on $\theta$. \hfill $\square$
\end{itemize}

\medskip

%\noindent \emph{Proof of Theorem \ref{res1}}.\\
%The proof of Theorem \ref{res1} is here divided in three steps: (a) recursive feasibility of the optimization problem \eqref{eq:optprb}, (b) constraints \eqref{eq:constraints_YU} satisfaction, (c) convergence.\\
%%\begin{itemize}
%	\item Proof of recursive feasibility of the optimization problem \eqref{eq:optprb}.
%	\item Proof that constraints \eqref{eq:constraints_YU} are satisfied.
%	\item Proof of convergence.
%\end{itemize}
%
\vspace{2 mm}
\noindent
%\emph{Recursive feasibility.}\\
We first prove recursive feasibility by induction. Assume that, at $k=j\bar{p}$, a solution to the optimization problem \eqref{eq:optprb} exists and
denote it with $\bar{X}^*(j|j)$, $\bar{\mathbf{U}}^*(j|j)$. All constraints \eqref{eq:constraints} are therefore verified by the nominal state trajectories associated with the optimal solution $\bar{X}^*(j+i|j)$ and $\mathbf{\bar{U}}^*(j|j)$:
\begin{subequations}
	\label{eq:constraints_k}
	\begin{align}
	X(j)-\bar{X}^*(j|j) &\in \mathbb{E} \label{constr:state} \\
	\bar{U}^*(j+i|j)&\in \mathbb{\mathbf{U}}\ominus K\mathbb{E}\label{constr:in}\\
	\hat{\bar{Z}}^*(j+i|j)&=\bar{C}\bar{X}^*(j+i|j) + \bar{D}\bar{U}^*(j+i|j)  \nonumber \\
	&\in \hat{\mathbb{\mathbf{Z}}}\ominus (\bar{C}+\bar{D}K)\mathbb{E}\label{constr:out}\\
	\bar{X}^*(j+N_p|j)&\in\mathbb{X}_F\label{constr:fin}
	\end{align}
\end{subequations}
with $i=0,\dots,N_p-1$. Finally, the input $U(j)$ is defined according to \eqref{eq:controllaw} with $\bar{U}(j)=\bar{U}^*(j|j)$ and $\bar{X}(j)=\bar{X}^*(j|j)$. Let us call this quantity $U^*(j|j)$. At yime $k=(j+1)\bar{p}$, 
$$X(j+1)=\bar{A}X(j)+\bar{B}U^*(j|j)+\bar{M}W(j)$$
We can show that a feasible, although possibly suboptimal, solution to \eqref{eq:optprb} can be defined, i.e., as $\bar{X}^*(j+1|j),\mathbf{\tilde{\bar{U}}}^*(j+1|j)=\left( \bar{U}^*(j+1|j), \dots, \bar{U}^*(j+N_p-1|j), K\bar{X}^*(j+N_p|j) \right) $.

First of all, we have
\begin{align*}
&X(j+1)-\bar{X}^*(j+1|j)=\\
&(\bar{A}+\bar{B}K)(X(j)-\bar{X}^*(j|j)) + \bar{M}W(j) \in \mathbb{E}
\end{align*}
in view of \eqref{constr:state} and of the fact that $\mathbb{E}$ is RPI.\\
Moreover, $\bar{U}^*(j+i|j) \in \mathbb{\mathbf{U}} \ominus K\mathbb{E}$ in view of \eqref{constr:in}, for all $i=1,\dots,N_p-1$, and $K\bar{X}^*(j+N_p|j) \in K\mathbb{X}_F \subseteq \mathbb{\mathbf{U}} \ominus K\mathbb{E}$ in view of \eqref{constr:fin} and of \eqref{constraint:Utightened}. 
In addition, $$(\bar{C}+\bar{D}K)\bar{X}^*(j+i|j) \in \hat{\mathbf{Z}} \ominus (\bar{C}+\bar{D}K)\mathbb{E}$$ for all $i=1,\dots,N_p-1$ in view of \eqref{constr:out} and $$(\bar{C}+\bar{D}K)\bar{X}^*(j+N_p|j) \in (\bar{C}+\bar{D}K)\mathbb{X}_F \subseteq \hat{\mathbf{Z}} \ominus (\bar{C}+\bar{D}K)\mathbb{E}$$ in view of \eqref{constr:fin} and of \eqref{constraint:Utightened}.

Finally, it holds that $$\bar{X}^*(j+N_p+1|j)=(\bar{A}+\bar{B}K)\bar{X}^*(j+N_p|j) \in \mathbb{X}_F$$ in view of \eqref{constr:fin} and of the positive invariance of $\mathbb{X}_F$.
Since feasibility holds by assumption at time $j\bar{p}, j=0$ then, by induction, it is guaranteed also for all $j>0$.

\vspace{2 mm}
\noindent
\emph{Constraint satisfaction.}\\
Constraint satisfaction is now proven. In view of the feasibility of the problem \eqref{eq:optprb} at any time instant $j\bar{p}, j\geq0$, it results that constraints \eqref{eq:constraints_k} are verified. Therefore, from \eqref{eq:controllaw}, \eqref{constr:state}, and \eqref{constr:in}, $$U^*(j|j)=\bar{U}^*(j|j)+K(X(j)-\bar{X}^*(j|j))\in ({\mathbb{\mathbf{U}}}\ominus K\mathbb{E})\oplus K\mathbb{E} \subseteq \mathbb{\mathbf{U}}.$$
Then, by definition of $U(j)$ and the set $\mathbb{\mathbf{U}}$, input constraints satisfaction in \eqref{eq:constraints_YU} follows.
Also, from \eqref{constr:state} and \eqref{constr:out},
\begin{align*}
\hat{Z}(j|j)&= \bar{C}\bar{X}^*(j|j)+\bar{D}U^*(j|j)+(\bar{C} + \bar{D}K)(X(j)-\bar{X}^*(j|j))\\
&\in (\hat{\mathbb{\mathbf{Z}}}\ominus (\bar{C}+\bar{D}K)\mathbb{E})\oplus  (\bar{C}+\bar{D}K)\mathbb{E}\subseteq \hat{\mathbb{\mathbf{Z}}}.
\end{align*}
Indeed, if $\hat{Z}(j|j)\in \hat{\mathbb{\mathbf{Z}}}$, then $Z(j|j) \in \mathbb{\mathbf{Z}}$ follows, and eventually by definition of the latter ones, \eqref{eq:constraints_YU} is satisfied.

\vspace{2mm}
\noindent
Convergence is proven with standard arguments (see \cite{rawlings2009model}) by showing that the optimal cost function is decreasing in time, i.e.
\begin{equation}
J^*(j+1|j+1)-J^*(j|j) \leq - \left(\|\hat{\bar{Z}}^*(j|j)\|_Q^2+\|\hat{\bar{U}}^*(j|j)\|_R^2 \right)
\label{J:decrease}
\end{equation}

Since $J(j)$ is positive by definition, and decreasing in view of \eqref{J:decrease}, then $\hat{\bar{Z}}^*(j|j) \text{ and } \bar{U}^*(j|j) \to 0$ as $j \to +\infty$.
Also, recalling \eqref{constr:state}, it holds
\begin{align*}
\bar{C}(X(j)-\bar{X}^*(j|j))+\bar{D}(U^*(j|j)-\bar{U}^*(j|j))&=\hat{Z}(j|j)-\hat{\bar{Z}}^*(j|j) \\
&\in (\bar{C}+\bar{D}K)\mathbb{E},
\end{align*} for all $j$ which concludes the proof. \hfill $\square$
\bibliography{mybiblio}
\bibliographystyle{plain}
\end{document}